%% file: main.tex
\documentclass[aps, prx, twocolumn]{revtex4-1} 
\usepackage[ colorlinks = true,
             linkcolor = blue,
             urlcolor  = blue,
             citecolor = blue,
             anchorcolor = green,
]{hyperref}
\usepackage{tabularx}
\usepackage{amsfonts}
\usepackage{amssymb}
\usepackage{amsthm}
\usepackage{amsmath}
\usepackage{amsthm}
\usepackage{amsxtra}
\usepackage{algpseudocode}
\usepackage{graphics}
\usepackage{enumerate}
\usepackage[utf8]{inputenc}
\def\bra #1{\langle #1\vert}
\def\ket #1{\vert #1\rangle}

\def\tr{{\rm Tr}}

\newcommand{\be}{\begin{equation}}
\newcommand{\ee}{\end{equation}}
\newcommand{\ba}{\begin{array}}
\newcommand{\ea}{\end{array}}
\newcommand{\bea}{\begin{eqnarray}}
\newcommand{\eea}{\end{eqnarray}}

\newcommand{\calL}{{\cal L }}

\newcommand{\calR}{{\cal R }}

\newcommand{\calC}{{\cal C }}

\newcommand{\calM}{{\cal M }}

\newcommand{\EE}{\mathbb{E}}

\newcommand{\RR}{\mathbb{R}}
\newcommand{\la}{\langle}
\newcommand{\ra}{\rangle}

\newcommand{\trace}[1]{{\mathrm{Tr}{#1}}}

\newcommand*{\ExpE}{\mathbb{E}}

\newcommand{\Lslater}[1]{{\lambda_{\mathrm{Slater}}({#1})}}
\newcommand{\Lgauss}[1]{{\lambda_{\mathrm{Gauss}}({#1})}}
\newcommand{\Lmax}[1]{{\lambda_{\max}({#1})}}

\newtheorem{lemma}{Lemma}
\newtheorem{fact}{Fact}

\newtheorem{theorem}{Theorem}
\newtheorem{problem}{Problem}

\begin{document}
\title{Approximation algorithms for quantum many-body problems}
\author{Sergey Bravyi$^1$}
\author{David Gosset$^{1,2}$}
\author{Robert K\"{o}nig$^3$}
\author{Kristan Temme$^1$}
\affiliation{$^{1}$IBM T.J. Watson Research Center, Yorktown Heights, NY 10598, USA}

\affiliation{$^{2}$ Department of Combinatorics \& Optimization and Institute for Quantum Computing, University of Waterloo, Waterloo, ON, N2L 3G1, Canada}
\affiliation{$^{3}$Institute for Advanced Study \& Zentrum Mathematik, Technical University of Munich, 85748 Garching, Germany}
\begin{abstract}
We discuss classical algorithms  for approximating the largest eigenvalue of 
quantum spin and fermionic Hamiltonians based on semidefinite programming relaxation methods.
First, we consider traceless $2$-local Hamiltonians $H$ describing a system of $n$ qubits.
We give an efficient algorithm that outputs a separable state whose energy
is at least $\lambda_{\max}/O(\log{n})$, where $\lambda_{\max}$ is the maximum
eigenvalue of $H$.   We also give a simplified proof of a theorem due to Lieb  that 
establishes the existence of a separable state with energy at least $\lambda_{\max}/9$.
Secondly, we consider a system of $n$ fermionic modes and  traceless
Hamiltonians composed of quadratic and quartic fermionic operators.
We give an efficient algorithm that outputs a   fermionic Gaussian state
whose energy is at least $\lambda_{\max}/O(n\log{n})$.
Finally, we show that  Gaussian states can vastly outperform Slater determinant
states commonly used in the Hartree-Fock method.
We give a simple family of Hamiltonians  for which
Gaussian states and Slater determinants approximate
$\lambda_{\max}$ within a fraction $1-O(n^{-1})$ and $O(n^{-1})$
respectively.
\end{abstract}
\maketitle

\section{Introduction}
\label{sec:intro}

Quantum many-body systems with local interactions are central to condensed matter physics and chemistry. Their significance in quantum computer science derives from the fact that computing the minimal or maximal energy configuration of such a system is a quantum analogue of constraint satisfaction \cite{kitaev2002classical, kempe2004complexity}. While the worst-case hardness of such quantum constraint satisfaction problems is well-understood (and largely parallels the classical theory) \cite{kitaev2002classical, kempe2004complexity, cubitt2016complexity, gharibian2015quantum}, the study of quantum approximation problems has been a topic of recent interest, motivated by the prospect of generalizing the classical PCP (probabilistically checkable proofs) theorem \cite{gharibian2012approximation, aharonov2013guest, brandao2013product, eldar2015local, nirkhe2018approximate}. 
Here  we show that  
optimization problems encountered in quantum many-body physics can
be tackled using approximation algorithms based on the semidefinite programming
relaxation method
pioneered by Goemans and Williamson~\cite{goemans1995improved}
and generalized  further in Refs.~\cite{charikar2004maximizing,nemirovski2007sums,so2009improved,so2011moment}.

Our starting point is the classical problem of maximizing a binary quadratic function
\begin{equation}
F(x)=x^T B x +v^T x,\qquad \quad x\in \{\pm 1\}^n,
\label{eq:f}
\end{equation}
defined by a matrix $B\in \mathbb{R}^{n\times n}$ 
and a vector $v\in \mathbb{R}^n$. 
We shall assume that $B$ has zero diagonal so that $F(x)$ has no constant terms.
Computing the maximum 
\[
F_{\max}=\max_{x\in \{\pm 1\}^n} F(x)
\]
exactly is NP-hard; for example, if $B$ is a $\{0,1\}$ matrix and $v=0$ then computing $F_{\max}$ is equivalent to computing the Max-Cut of the simple graph with adjacency matrix $B$. Charikar and Wirth \cite{charikar2004maximizing} considered the approximation problem in which one aims to compute $x\in \{\pm 1\}^n$ such that  the approximation ratio $F(x)/F_{max}$ is as large as possible.  They showed that an efficient classical algorithm based on rounding a semidefinite programming relaxation achieves an approximation ratio of $\Omega(\log^{-1}(n))$.  Conversely, Arora et al. \cite{arora2005non} have established that for some absolute constant $0\leq \gamma\leq 1$  it is quasi-NP hard to obtain a $\Omega(\log^{-\gamma}(n))$ approximation ratio \footnote{In particular, Ref.~\cite{arora2005non} shows that if an efficient algorithm achieves this approximation ratio then there exists an algorithm which solves any decision problem in $NP$ on input size $n$ using runtime $n^{\mathrm{poly}(\log(n))}$. This is believed to be very unlikely.}.  We note that including a linear term in Eq.~\eqref{eq:f} is unnecessary, as there is a simple and efficient reduction to the case $v=0$~\footnote{Given a function Eq.~\eqref{eq:f} we can add an auxiliary variable $y\in \{\pm1\}$ and consider the function of $n+1$ variables $G(x,y)=x^T B x+y b^T x$ which only contains quadratic terms. It is then easily seen that the range of $G$ is equal to the range of $F$.}.
On the other hand, the definition of approximation ratio used here depends crucially on the
assumption that $F(x)$ has no constant terms (in particular, $F_{max}\ge 0$ since
the expected value of $F(x)$ on a random uniform bit string $x$ is zero).

In the present paper we consider a natural quantum analogue of binary quadratic functions
-- traceless Hamiltonians $H$ that describe systems of qubits or fermions
with two-body interactions. We show how to adapt approximation algorithms
developed in the classical case to approximate the maximum (or minimum) 
eigenvalue of $H$. 
We discuss qubit Hamiltonians and approximations by separable states
in Section~\ref{sec:qubit}. Fermionic 
Hamiltonians and approximations based on Slater determinants
and Gaussian states are discussed 
 in Section~\ref{sec:fermion}.

\section{Two-local qubit Hamiltonians}
\label{sec:qubit}

A traceless $2$-local  Hamiltonian is a quantum generalization of the binary quadratic function Eq.~\eqref{eq:f}. Write the Pauli operators acting on the $a$-th qubit as $P_{3a-2}=X_a, P_{3a-1}=Y_a, P_{3a}=Z_a$.
Here $1\le a\le n$. We shall consider traceless $2$-local Hamiltonians acting on a system of $n$ qubits, that is
\begin{align}
H&=H_1+H_2\label{eq:Htraceless} \\
H_1&=\sum_{j=1}^{3n} D_jP_j\ \qquad  H_2 =\sum_{i,j=1}^{3n}C_{i,j}P_i P_j\nonumber
\end{align}
where we assume that $C_{i,j}=0$ if $P_i$ and $P_j$ act on the same qubit. Without loss of generality $C^T=C\in\mathbb{R}^{3n\times 3n}$ is symmetric.  Note that any traceless $2$-local Hamiltonian can be expressed as in Eq.~\eqref{eq:Htraceless}. Moreover, the classical binary quadratic optimization problem described above is obtained as a special case where the Hamiltonian is diagonal in the computational basis.

The maximum energy of $H$ is its largest eigenvalue~\footnote{We note that all our results apply also
to the problem of minimizing the energy of $H$ and approximating the minimum eigenvalue 
$\lambda_{\min}(H)=-\lambda_{\max}(-H)$ which is more relevant 
in many-body physics.
We opted to consider a maximization problem to avoid a proliferation of  minus signs.}
which we denote $\lambda_{\max}(H)$. Computing this quantity exactly or
estimating it  with a small additive error is 
known to be a  QMA-hard problem~\cite{kempe2004complexity}.
It is considered highly unlikely that such problems 
admit an efficient algorithm, either classical or quantum.
Instead, here we seek to compute an estimate $\tilde{\lambda}\leq \lambda_{\max}(H)$ with a good approximation ratio $\tilde{\lambda}/\lambda_{\max}(H)$. What is the best (largest) approximation ratio that can be achieved by a polynomial-time algorithm? The hardness result of Arora et al.~ \cite{arora2005non} imply that we cannot hope to beat $\Omega(\log^{-\gamma}(n))$ for some constant $\gamma>0$. In this section we generalize the algorithm of Ref.~\cite{charikar2004maximizing} while retaining its performance guarantee.
\begin{theorem}
There is an efficient classical algorithm which, given $H$ of the form in Eq. \eqref{eq:Htraceless}, outputs a product state $|\phi\rangle=|\phi_1\rangle\otimes \ldots \otimes |\phi_n\rangle$ such that with probability at least $2/3$ 
\[
\langle \phi|H|\phi\rangle\geq \frac{\Lmax{H}}{O(\log{n})}.
\]
Moreover, each single-qubit state $\phi_i$ is an eigenstate of one of the Pauli operators $X,Y$ or $Z$.
\label{thm:alg}
\end{theorem}

The proof (given below) is based on rounding a semidefinite programming relaxation of the optimization over product states, mirroring the classical proof from Ref. ~\cite{charikar2004maximizing}.

Theorem \ref{thm:alg} complements previous work on product state approximations for local Hamiltonian problems and the fundamental limitations of mean-field theory \cite{gharibian2012approximation,brandao2013product,harrow2017extremal}. However, our setting is slightly different and our results are not directly comparable. For example, 
Gharibian and Kempe~\cite{gharibian2012approximation} 
studied approximation ratio with respect to the maximal eigenvalue of a local Hamiltonian which is a sum of positive semidefinite terms, whereas an essential feature of our definition is that the Hamiltonian is traceless.
Brandao and Harrow~\cite{brandao2013product} established upper bounds on the additive error between the energy attainable by a product state and the maximal eigenvalue. The most closely related result is an algorithm due to Harrow and Montanaro \cite{harrow2017extremal} which, given a traceless $2$-local Hamiltonian $H$ of the form in Eq.~\eqref{eq:Htraceless}, outputs a product state $|\phi\rangle$ with energy at least a $\Omega(n^{-1})$ fraction of the $1$-norm of the coefficients appearing in the Hamiltonian, i.e., 
\begin{equation}
\langle \phi|H|\phi\rangle \geq \frac{1}{Kn} \left( \sum_{i,j=1}^{3n} |C_{ij}|+\sum_{i=1}^{3n} |D_i|\right),
\label{eq:hm}
\end{equation}
where $K>0$ is an absolute constant. Here we study a different notion of approximation ratio defined with respect to the maximum eigenvalue $\lambda_{\max}(H)$ rather than the $1$-norm of the coefficients.

While we do not expect a polynomial-time algorithm to significantly beat the approximation ratio achieved by Theorem \ref{thm:alg}, it is natural to ask: what is the best possible approximation ratio that is achieved by a product state (even if such state cannot be efficiently found) ?  Define
\[
\lambda_{\mathrm{sep}}(H)=\max_{\phi_1,\ldots, \phi_n}  \quad \langle \phi_1\otimes \phi_2 \otimes \ldots \phi_n |H| \phi_1\otimes \phi_2 \otimes \ldots \phi_n \rangle.
\]
where the maximization is over normalized single-qubit states $\phi_1,\ldots, \phi_n$. A counterpart to Theorem \ref{thm:alg} and the quasi-NP hardness result of Ref.~\cite{arora2005non} is that there always exists a product state whose energy achieves a constant approximation ratio.
In particular, 
using a quasi-classical representation of quantum spin systems Lieb~\cite{lieb1973classical}
established the following
result~\footnote{The statement of 
Theorem~\ref{thm:lieb} with $H_1=0$ follows from Eq.~(1.1) of Ref.~\cite{lieb1973classical}
by considering spin-$1/2$ particles and taking the zero temperature limit.}.
\begin{theorem}[\bf Lieb 1973]
Suppose $H$ is a traceless $2$-local Hamiltonian. Then 
\begin{equation}
\lambda_{\mathrm{sep}}(H)\geq \frac{1}{9} \lambda_{\max}(H).
\label{eq:constantratio}
\end{equation}
\label{thm:lieb}
\end{theorem}
Ref.~\cite{lieb1973classical} proved the theorem
for $2$-local Hamiltonians without linear terms (that is, $H_1=0$).
Here we remove this restriction and give a simplified proof
of Lieb's theorem which is based on applying an entanglement-breaking depolarizing channel to each qubit of an eigenvector of $H$ with maximal eigenvalue. 
We also establish that the above approximation ratio is achieved by a tensor product of single-qubit stabilizer states, i.e., a state $|\phi\rangle=|\phi_1\rangle\otimes \ldots \otimes |\phi_n\rangle$  where each $|\phi_i\rangle$ is an eigenstate of one of the Pauli operators $X,Y$ or $Z$ \footnote{The fact that the product states in Theorems \ref{thm:alg} and \ref{thm:lieb} can be taken to have this special form is a consequence of the fact that the six eigenstates of single-qubit Pauli operators $\{ X,Y,Z\}$ form a $2$-design \cite{dankert2009exact}. Any other single-qubit $2$-design could alternatively be used in its place, such as for example the $4$-state one which is used for similar purposes in Ref. \cite{harrow2017extremal}}.

We suspect that the constant $1/9$ appearing in Lieb's theorem is not optimal and leave this as an open question. 
By minimizing the approximation ratio $R(H)=\lambda_{\mathrm{sep}}(H)/ \lambda_{\max}(H)$
numerically over the set of all $2$-local Hamiltonians $H$ we observed that
$\min_H R(H)\approx 1/3$ for $n\le 6$ qubits. 
Along these lines we establish that product states achieve a slightly better constant approximation ratio for the related problem of maximizing the \textit{absolute value} of the energy. In particular, defining $\|H\|_{\mathrm{sep}}=\max \left\{\lambda_{\mathrm{sep}}(H), \lambda_{\mathrm{sep}}(-H)\right\}$, we prove 
\begin{equation}
\|H\|_{\mathrm{sep}}\geq \frac{1}{6}\|H\|.
\label{eq:normbnd}
\end{equation}
The proof, provided in the Appendix~\ref{app:Lieb}, is a variant of a strategy used by H\aa stad to establish approximation guarantees for classical constraint satisfaction problems~\cite{hastad2015improved}, augmented with some extra quantum ingredients such as an entanglement-breaking measurement.

We begin with a lemma that allows us to reduce Theorems \ref{thm:alg}, \ref{thm:lieb} to the special case where the linear term vanishes, i.e., $H_1=0$. Given a general Hamiltonian of the form Eq.~\eqref{eq:Htraceless}, we form the $n+1$-qubit Hamiltonian
\[
H^{\prime}=H_2+Z_{n+1} H_1
\]
which has no linear term. 
\begin{lemma}
$\lambda_{\max}(H^{\prime})=\lambda_{\max}(H)$. Moreover, given any $(n+1)$-qubit product state $\omega$ we can efficiently compute an $n$-qubit product state $\phi$ such that 
\begin{equation}
\langle \phi|H|\phi\rangle\geq \langle \omega |H^{\prime}|\omega\rangle.
\label{eq:energyatleast}
\end{equation}
If $\omega$ is a tensor product of single-qubit stabilizer  states  then so is $\phi$.
\label{lem:timereversal}
\end{lemma}
\begin{proof}
Since $Z_{n+1}$ commutes with $H^{\prime}$, all eigenvalues of $H^{\prime}$ are either eigenvalues of $H_1+H_2$ or $H_2-H_1$. The operator $H_2-H_1$ is obtained from $H$ by the \textit{time-reversal map}: 
\begin{equation}
(Y^{\otimes n} (H_2+H_1 )Y^{\otimes n})^T=H_2-H_1
\label{eq:timerev}
\end{equation}
where $^T$ indicates the matrix transpose (in the computational basis). Since conjugation by the unitary operator $Y^{\otimes  n}$ and the matrix transpose operation both preserve the spectrum, we see that $H_2-H_1$ and $H_2+H_1$ have the same eigenvalues, and thus so do $H, H^{\prime}$. 

Now suppose we are given an $n+1$-qubit product state $|\omega\rangle=|\omega_1\otimes \ldots \otimes \omega_{n+1}\rangle$. Since $Z_{n+1}$ commutes with $H^{\prime}$, one of the product states
\[
|\omega(z)\rangle=|\omega_1\otimes \ldots \otimes \omega_{n}\rangle\otimes |z\rangle \qquad z\in \{0,1\}
\]
has energy at least that of $\omega$.  If $z=0$ we take $|\phi=|\omega_1\otimes \ldots \otimes \omega_{n}\rangle$ while if $z=1$ we take
\[
|\phi\rangle=Y^{\otimes n}|\omega_1^{\star}\otimes \ldots \otimes \omega_{n}^{\star}\rangle
\]
where $^{\star}$ indicates the complex conjugate in the computational basis. One can then directly check (using Eq.~\eqref{eq:timerev} for the case $z=1$) that $\phi$ satisfies Eq.~\eqref{eq:energyatleast}.
\end{proof}
\begin{proof}[\bf Proof of Theorem \ref{thm:lieb}]
We fix $H_1=0$ below; Lemma \ref{lem:timereversal} implies this is without loss of generality. For any $\delta\in [0,1]$ let $\mathcal{E}_{\delta}$ be a single-qubit depolarizing channel defined by its action on the basis $\{I,X,Y,Z\}$ of $2\times 2$ Hermitian operators:
\begin{equation}
\mathcal{E}_{\delta}(I)=I \qquad \mathcal{E}_{\delta}(P)=\delta P \qquad P\in \{X,Y,Z\}.
\label{eq:depdef}
\end{equation}
When $\delta=1$ we recover the identity channel, and when $\delta=0$ the channel is maximally noisy. It is well known that below a critical value $\delta\leq \frac{1}{3}$ the depolarizing channel is entanglement breaking \cite{horodecki2003entanglement, ruskai2003qubit}. As a consequence, for any $n$-qubit state $\rho$, the depolarized state $\sigma=\mathcal{E}_{\frac{1}{3}}^{\otimes n} (\rho)$ is separable (i.e, a mixture of product states). One way to see this explicitly is to use the identity
\begin{equation}
\mathcal{E}_{\frac{1}{3}}(R)=\frac{1}{3} \sum_{Q\in \{\pm X,\pm Y,\pm Z\}}  \langle Q|R |Q\rangle \;  |Q\rangle\langle Q|.
\label{eq:Q}
\end{equation}
Here $|\pm Q\rangle$ is the eigenvector of the single-qubit Pauli matrix $Q$ with eigenvalue $\pm 1$. Eq. ~\eqref{eq:Q} shows that a state 
\[
\sigma=\mathcal{E}_{\frac{1}{3}}^{\otimes n}(\rho)
\] 
can be prepared by measuring each qubit of $\rho$ in the Pauli $X,Y$ or $Z$ basis uniformly at random, and is therefore clearly separable. Indeed, $\sigma$ is always a probabilistic mixture of product states $|\phi_1\rangle\otimes\ldots \otimes |\phi_n\rangle$ such that each $\phi_i$ is an eigenstate of either $X,Y$ or $Z$. Moreover, using Eq. \eqref{eq:depdef} we see that expectation values in the state $\sigma$ are simply related to those of $\rho$:
\begin{equation}
\mathrm{Tr}(\sigma P_{j_1} P_{j_2}\ldots P_{j_L})=\frac{1}{3^L}\mathrm{Tr}(\rho P_{j_1} P_{j_2}\ldots P_{j_L})
\label{eq:rescale}
\end{equation}
for $L$ Pauli operators acting on distinct qubits. 

Now consider an $n$-qubit state $\psi$ satisfying $\langle \psi|H|\psi\rangle=\lambda_{\max}(H)$. By the above argument the depolarized state 
\[
\sigma=\mathcal{E}_{\frac{1}{3}}^{\otimes n}(|\psi\rangle\langle\psi|)
\]
is separable and using Eq.~\eqref{eq:rescale} gives
\begin{equation}
\lambda_{\mathrm{sep}}(H)\geq \mathrm{Tr}(\sigma H)=\frac{1}{9}\langle \psi|H| \psi\rangle
\end{equation}
\end{proof}
We prove Theorem \ref{thm:alg} following essentially the same steps used in  Ref.~\cite{charikar2004maximizing} for the classical case.
\begin{proof}[\bf Proof of Theorem \ref{thm:alg}]
Below we assume $H_1=0$ without loss of generality (due to Lemma \ref{lem:timereversal}).

Let $\mathrm{Herm}(m)$ be the set of $m\times m$ Hermitian matrices.
Consider the following semidefinite program:
\begin{align}
\label{SDP1}
&\mbox{maximize $\mathrm{Tr}(CM)$ over $M\in \mathrm{Herm}(3n)$}\\
\label{SDP2}
&\mbox{subject to $M\ge 0$ and $M_{i,i}=1$ for all $i$}.
\end{align}
The first step of the algorithm is to compute the optimal solution $M$, which can be done in polynomial time using standard techniques. Note that $M$ provides an upper bound 
\begin{equation}
\label{eq:relax1}
\mathrm{Tr}(CM) \ge \lambda_{max}(H).
\end{equation}
Indeed, if $\psi$ is a normalized $n$-qubit state such that 
$\lambda_{max}(H)=\langle \psi|H|\psi\rangle$
then the matrix $K$ with $K_{i,j}=\langle \psi|P_i P_j|\psi\rangle$ is a feasible solution of the SDP satisfying $\mathrm{Tr}(CK)=\lambda_{max}(H)$.

We may assume wlog that $M$ is a real matrix (otherwise, replace $M$ by $(M+M^*)/2$).
Then one can represent $M$ as 
\begin{equation}
\label{v}
M_{i,j}=\langle v^i|v^j\rangle
\end{equation}
for some unit vectors $v^1,v^2,\ldots,v^{3n+1}\in \RR^{3n+1}$. Let $c=O(1)$ be a constant to be chosen later. The algorithm proceeds as described in the following pseudocode.
\begin{center}
\fbox{\parbox{1\linewidth}{
\begin{algorithmic}
\State{$T\leftarrow c\sqrt{\log(n)}$}
\State{$\ket{r}\leftarrow $ vector of $3n$ i.i.d.~$N(0,1)$ random variables}
 \For{$i=1,\ldots,3n$}
  \State{$z_i\leftarrow \langle r|v^i\rangle/T$}
  \If {$|z_i|>1/\sqrt{3}$} 
      \State{$y_i\leftarrow \mathsf{sgn}(z_i)/\sqrt{3}$}
  \Else
     \State{$y_i\leftarrow z_i$}  
  \EndIf
 \EndFor\\
\Return{$\rho=\rho_1\otimes\cdots\otimes\rho_n$, where\\ $\rho_a=1/2 (I+y_{3a-2}P_{3a-2}+y_{3a-1}P_{3a-1}+y_{3a}P_{3a})$. }
\end{algorithmic}
}}
\end{center}
Consider the output state $\rho$ of this algorithm. Since $\rho$ is a product state, one has $\tr(\rho P_i P_j)=\tr(\rho_i P_i)\tr(\rho_jP_j)=y_iy_j$ if $P_i$ and $P_j$ act on different qubits. Therefore
\begin{align}
\tr(H\rho)&=y^T C y.\label{eq:quadraticpart}
\end{align}
Below we establish the following approximation guarantee :
\begin{lemma}
We may choose $c=O(1)$ such that
\begin{align}
\mathbb{E}_r\left[\tr(\rho H)\right] &\geq 
\frac{\lambda_{\max}(H)}{ O(\log{n})}. \label{eq:approximationratio}
\end{align}
\label{lem:approx}
\end{lemma}
This is almost what we want to prove. To go from Lemma~\ref{lem:approx} to Theorem \ref{thm:alg}, we need to show that we can efficiently compute a pure product state (of the claimed form) that achieves the same approximation ratio with high probability (rather than merely in expectation).

We shall use the entanglement breaking measurement from the proof of Theorem \ref{thm:lieb} to map the state $\rho$ to a tensor product of single-qubit stabilizer states. To this end let us choose a string $b\in \{X,Y,Z\}^{n}$ uniformly at random and then measure each qubit of $\rho$ in the corresponding single-qubit basis (more precisely, let us simulate this process using our classical description of $\rho$). Suppose the measurement outcome is $s\in \{\pm1\}^n$. The post-measurement state is a product state $|\phi\rangle=|\phi_1\rangle\otimes \ldots \otimes |\phi_n\rangle$ where each $\phi_i$ is an eigenstate of $X,Y$, or $Z$.   Moreover,
\[
\mathbb{E}_{b,s}[|\phi\rangle\langle \phi|]=\mathcal{E}_{1/3}^{\otimes n} (\rho)
\]
where $\mathcal{E}_{\delta}(\cdot)$ is the depolarizing channel defined in Eq.~\eqref{eq:depdef}. Using Eq.~\eqref{eq:rescale} we get
\begin{align}
\mathbb{E}_{b,s,r}\left[\langle \phi| H|\phi\rangle\right]&\geq \mathbb{E}_{r}\left[1/9\tr(\rho H)\right]\nonumber\\
&\geq \Omega(\log^{-1}(n))\frac{1}{9}\lambda_{\max}(H),
\label{eq:expectationbound}
\end{align}
where in the last line we used Eq.~(\ref{eq:approximationratio}). Now using the upper bound $\langle \phi| H|\phi\rangle\leq \lambda_{\max}(H)$ we get
\[
\mathrm{Pr}\left[\langle \phi|H|\phi\rangle \geq \alpha\right]\geq \frac{\mathbb{E}_{b,s,r}\left[\langle \phi| H|\phi\rangle\right]-\alpha}{\lambda_{\max}(H)-\alpha} 
\]
for all $\alpha\leq \lambda_{\max}(H)$. Choosing $\alpha=\frac{1}{2}\mathbb{E}_{b,s,r}\left[\langle \phi| H|\phi\rangle\right]$ and substituting Eq.~\eqref{eq:expectationbound} we see that with probability at least $\Omega(\log^{-1}(n))$ the product state $\phi$ satisfies $\langle \phi|H|\phi\rangle \geq \Omega(\log^{-1}(n))\lambda_{\max}(H)$. Repeating the above procedure $O(\log(n))$ times and choosing the output product state $\phi$ with the highest energy is therefore sufficient to increase the success probability above $2/3$.
\end{proof}

It remains to prove Lemma~\ref{lem:approx}. First we prove that the $y_i$ variables are a good approximation to the $z_i$ variables in the following sense.
\begin{lemma}
Define $\Delta_{i,j}=z_iz_j-y_iy_j$. Then
\begin{align}
\ExpE_r |\Delta_{i,j}|&\leq e^{-\Omega(T^2)}\qquad 0\leq i<j\leq 3n.
\label{eq:L2}
\end{align}
\label{lem:l2}
\end{lemma}
\begin{proof}
Define
\[
\chi(r)=\left\{ \ba{rcl} 
0 &\mbox{if} & \mbox{$y_i=z_i$ and $y_j=z_j$} \\
1 && \mbox{otherwise}\\
\ea
\right.
\]
Noting that $\Delta_{i,j}=\Delta_{i,j} \chi(r)$ and that $|\Delta_{i,j}|\le 2|z_i z_j|$ one gets
\begin{equation}
\label{L2p1}
\EE_r |\Delta_{i,j}| \le 2 \EE_r |z_i z_j| \chi(r).
\end{equation}
By Cauchy-Schwarz,
$\EE_r f(r) g(r)  \le \sqrt{\EE_r f^2(r)} \cdot \sqrt{\EE_r g^2(r)}$
for any real-valued  functions $f(r)$ and $g(r)$.
Choose $f(r)= |z_i z_j|$ and $g(r)=\chi(r)$.
Since $\|v^i\|=\|v^j\|=1$, random variables
$Tz_i$ and $Tz_j$ are normally distributed according to $N(0,1)$.
Hence 
\begin{equation}
\label{L2p2}
\EE_r z_i^2 z_j^2 \le \frac12 \left(  \EE_r z_i^4 + \EE_r z_j^4 \right) =O(T^{-4})=O(1).
\end{equation}
By the union  bound, 
\begin{equation}
\label{L2p3}
\EE_r \chi(r)^2=
\EE_r \chi(r) \le
 2\mathrm{Pr}(|z_i| >1/\sqrt{3} ) 
= e^{-\Omega(T^2)}.
\end{equation}
Here we again used the fact that $Tz_i\in N(0,1)$.
Combining all above proves Eq.~(\ref{eq:L2}).
\end{proof}

\begin{proof}[\bf Proof of Lemma \ref{lem:approx}]
Using Eq.~\eqref{eq:quadraticpart} we get
\begin{align}
\ExpE_r\left[\tr(H\rho)\right]&=\ExpE_r \left[y^TCy\right]\nonumber\\
&=\EE_r (z^T C z) - \EE_r \mathrm{Tr}(C\Delta)\\
&=\frac{\mathrm{Tr}(CM)}{T^2}  - \EE_r \mathrm{Tr}(C\Delta) \label{eq:energy},
\end{align}
where $\Delta$ is defined in Lemma~\ref{lem:l2}, and in the last line we used the fact that $\EE_r (z_iz_j)=\langle v^i|v^j\rangle/T^2$.

Substituting Eq.~\eqref{eq:relax1} in the first term of Eq.~\eqref{eq:energy} and upper bounding the second term one gets
\begin{align}
\ExpE_r\left[\tr(H\rho)\right] &\ge \frac{\lambda_{\max}(H)}{T^2}  - \sum_{i,j=1}^{3n} |C_{i,j}| \cdot \EE_r |\Delta_{i,j}|\\
& \ge \frac{\lambda_{\max}(H)}{T^2}-e^{-\Omega(T^2)}\sum_{i,j=1}^{3n}|C_{i,j}|\ .
\label{eq:almostdone}
\end{align}
where in the last line we used Lemma \ref{lem:l2}. We bound the second term using the fact that 
\begin{equation}
\sum_{i,j=1}^{3n} |C_{i,j}|\leq Kn \lambda_{\max}(H)
\label{eq:1norm}
\end{equation}
for some absolute constant $K>0$, which follows directly from Eq.~\eqref{eq:hm} and is proved in Ref. \cite{harrow2017extremal}. For completeness we provide an alternative proof of Eq.~\eqref{eq:1norm} in Appendix~\ref{app:LmaxLower}.
From Eqs.~(\ref{eq:almostdone}, \ref{eq:1norm}) we get
\[
\ExpE_r\left[\tr(H\rho)\right] \ge \lambda_{\max}(H)\left(\frac{1}{T^2}-e^{-\Omega(T^2)}K n\right),
\]
where $K=O(1)$. Now we can see that choosing $T=c\sqrt{\log(n)}$ with $c=O(1)$ is sufficient to ensure the term in parentheses is $\Omega(\log^{-1}(n))$.
\end{proof}

\section{Many-body fermionic problems}
\label{sec:fermion}

So far we have viewed a traceless $2$-local Hamiltonian as a quantum generalization of a binary quadratic function. Another physically motivated generalization is a system of fermionic modes with two-body interactions.
The Hilbert space of $n$ fermi modes can be identified with the one of $n$ qubits
and equipped with the standard basis $\{|x\ra\}$, where $x\in \{0,1\}^n$.
Here $x_j=0$ or $x_j=1$ indicate that the $j$-th mode is empty or occupied 
by a fermionic particle.
Define particle annihilation operators $a_1,\ldots,a_n$ such that 
$a_1=|0\ra\la 1|_1$ and 
\[
a_j=Z_1\cdots Z_{j-1} |0\ra\la 1|_j
\]
for $2\le j\le n$. The corresponding creation operators
are defined as $a_1^\dag,\ldots,a_n^\dag$.
They obey commutation rules
\[
a_i a_j=-a_j a_i \quad \mbox{and} \quad a_i^\dag a_j + a_j a_i^\dag = \delta_{i,j} I
\]
for all $i,j$.  Most of the fermionic systems studied in physics can be described by a Hamiltonian
\begin{align}
h&= h_1+h_2 + \omega I
\label{fermiH12}\\
h_1 &= \sum_{p,q=1}^n V_{p,q} a_p^\dag a_q, \qquad
h_2 = \sum_{p,q,r,s=1}^{n} W_{pqrs} a_p^\dag a_q^\dag a_r a_s. \nonumber
\end{align}
Here $V_{p,q}$ and 
$W_{pqrs}$ are complex coefficients chosen such that $h_1$ and $h_2$ are hermitian. 
The last term is an overall energy shift $\omega\in \RR$.
We note that the quadratic term $h_1$ is the fermionic analogue of a $1$-local Hamiltonian
for qubits. In particular, one can  map $h_1$
to a linear combination of single-mode operators $a_j^\dag a_j=(I-Z_j)/2$
by performing a suitable change of basis, see e.g. Ref. \cite{terhal2002classical}.
The quartic term $h_2$ allows for non-trivial two-body interactions between fermions. For example, a notable special case of Eq.~\eqref{fermiH12} are the Hamiltonians which describe molecular structure in quantum chemistry. 

As before, we are interested in approximating the maximum eigenvalue of $h$ denoted 
$\lambda_{\max}(h)$.  Can we match the approximation guarantees of Theorems~\ref{thm:alg},\ref{thm:lieb}
with  appropriate fermionic analogs of product states? 
Natural candidates are Slater determinant states used in the Hartree-Fock method from quantum chemistry.
Recall that a Slater determinant state~$\psi$ can be specified by the number of
particles $0\le k\le n$ and a unitary matrix $U$ of size $n$ such that 
\begin{align}
\ket{\psi}=b_1^\dag\cdots b_k^\dag\ket{0^n}\qquad\textrm{ with }\qquad b_p=\sum_{q=1}^n U_{p,q}a_q.
\end{align}
Let 
\be
\label{LSlater}
\Lslater{h} = \max_{\psi} \la \psi |h|\psi\ra,
\ee
where the maximization is over all Slater determinant states $\psi$ (with $k=0,1,\ldots,n$).
We can establish the following weaker version  of Theorem~\ref{thm:alg}. 
\begin{theorem}
There is a classical algorithm which takes as input a traceless Hamiltonian $h$ of the form Eq.~\eqref{fermiH12} and
outputs a Slater determinant state~$\psi$ such that 
\be
\label{SlaterPart1}
\langle\psi|h|\psi\rangle \ge \frac{\Lslater{h}}{O(\log{n})}
\ee
with probability at least $2/3$. The algorithm has runtime
$poly(n)$.
\label{thm:slater}
\end{theorem}
The approximation ratio in    Eq.~(\ref{SlaterPart1}) is close to optimal. 
Indeed,  fermionic Hamiltonians of the form Eq.~(\ref{fermiH12}) subsume
classical quadratic functions $F(x)$ defined in Eq.~(\ref{eq:f}). 
The latter can be  
expressed by a diagonal Hamiltonian $h$
such that $h|x\ra=F(x)|x\ra$ for all $x$.
Then $\Lslater{h}=\Lmax{h}=F_{max}$.
Moreover, $h$ is  quadratic in Pauli operators
$Z_j=a_j a_j^\dag - a_j^\dag a_j$, that is,
$h$ has the form Eq.~(\ref{fermiH12}).
The classical hardness result~\cite{arora2005non} 
then implies that improving the approximation ratio in Eq.~(\ref{SlaterPart1})
beyond $\Omega(\log^{-\gamma}(n))$ is  quasi-NP hard 
for some $\gamma>0$.

The proof of Theorem~\ref{thm:slater}, given in Appendix~\ref{app:SlaterProof},
relies on the fact that the optimization problem defining  $\Lslater{h}$ can be
rephrased as a  quadratic optimization with orthogonality constraints (known as Qp-Oc)~\cite{nemirovski2007sums,so2009improved,so2011moment}.
The latter admits an efficient approximation algorithm based on 
a suitable SDP-type relaxation~\cite{so2009improved} similar to the one
considered in Section~\ref{sec:qubit}.
For completeness, we provide all relevant facts regarding Qp-Oc in
Appendix~\ref{app:QPOC}.

Theorem~\ref{thm:slater} motivates the question of how well
$\Lslater{h}$ approximates the largest eigenvalue $\Lmax{h}$
and whether one can establish a fermionic analogue of 
Lieb's theorem for Slater determinants.
Unfortunately,  we show that the   ratio
$\Lslater{h}/\Lmax{h}$ can be as small as $O(n^{-1})$.
Namely, we 
present  a family of traceless fermionic Hamiltonians  $h$
of the form Eq.~\eqref{fermiH12}
for which $\lambda_{\max}(h)$ scales as $\Omega(n^2)$
whereas the energy of any 
Slater determinant state is at most $O(n)$. 
To state this result consider a variant  of  Richardson's Hamiltonian \cite{richardson1965exact,richardson} 
defined as
\begin{align}
h=P^\dagger P \qquad \qquad P=\sum_{j=1}^{n/2} a_{2j-1}a_{2j}.\label{eq:richardson}
\end{align}
Here we assume that the number of modes $n=2N$ is even.
Note that $h$ has the form Eq.~(\ref{fermiH12}) with $h_1=0$. Richardson's model is exactly solvable~\cite{richardson1965exact}.
In particular, 
\be
\label{Lmax}
\lambda_{\max}(h) = \frac{N(N+2) + \epsilon_N}4,
\ee
where $\epsilon_N={N \pmod 2} \in \{0,1\}$.
For completeness, we provide a simple proof of Eq.~(\ref{Lmax}) in Appendix~\ref{app:Richardson}.
\begin{lemma}\label{lemma:slaterlimitation}
Let $h$ be Richardson's Hamiltonian 
acting on $n$ fermi modes. Then 
\be
\label{SlaterUpper}
 \frac{\Lslater{h}}{\Lmax{h}}\leq \frac{8}{n}.
\ee
\end{lemma}
We note that Richardson's Hamiltonian $h$ can be made traceless by
performing an energy shift $h\gets h- (n/8) I$. Since $\lambda_{\max}(h)$ 
is proportional to $n^2$, 
such an energy shift does not affect the conclusion that
Slater determinants
achieve approximation ratio at most $O(n^{-1})$.
A physical intuition behind Lemma~\ref{lemma:slaterlimitation}
comes from the fact that Slater determinants
cannot describe states with a  superconducting
order parameter.
At the same time, Richardson's Hamiltonian
Eq.~(\ref{eq:richardson})
describes a system of fermions with attractive 
interactions~\footnote{Recall that the sign of $h$ has to be flipped if one is interested in the 
minimum rather than maximum eigenvalue.}
which favor a superconducting order.

\begin{proof}[\bf Proof of Lemma~\ref{lemma:slaterlimitation}]
Let $\psi$ be a  Slater determinant state with $k$ particles.
Define a covariance matrix
\[
Q_{ij}=\langle \psi |a_i^{\dagger} a_j|\psi\rangle \qquad 1\leq i,j\leq n.
\]
One can easily check that 
$Q$ is a rank-$k$ projector.
The fermionic version of Wick's theorem asserts that 
\[
\la \psi |(a_p a_q)^\dag a_r a_s |\psi\ra =
Q_{p,r} Q_{q,s} - Q_{p,s} Q_{q,r}
\]
for any tuple of modes $p,q,r,s$.
This gives
\begin{align*}
\la \psi|h|\psi\ra
&=\sum_{i,j=1}^{N} \langle \psi|a_{2i}^{\dagger} 
a_{2i-1}^{\dagger} a_{2j-1}a_{2j}|\psi\rangle\\
&=\sum_{i,j=1}^{N} Q_{2i,2j}Q_{2i-1,2j-1}-Q_{2i,2j-1}Q_{2i-1,2j}.
\end{align*}
The inequality $2|ab|\leq |a|^2+|b|^2$ gives
\begin{align}
\la \psi|h|\psi\ra
&\le \frac{1}{2}\sum_{i,j=1}^{N}\big(|Q_{2i,2j}|^2+|Q_{2i-1,2j-1}|^2\nonumber\\&+|Q_{2i,2j-1}|^2+|Q_{2i-1,2j}|^2\big)\nonumber\\
&=\frac{1}{2}\trace(Q^{\dagger}Q)=\frac12\trace{(Q)} = k/2 \le N.
\end{align}
where in the last line we used the fact that $Q$ is rank-$k$ projector
and $k\le n=2N$.  Since this is true for any Slater determinant $\psi$, we arrive at
\begin{align}
\Lslater{h}\le N.\label{eq:slatersecond}
\end{align}
Combining Eqs.~(\ref{Lmax},\ref{eq:slatersecond}) proves the lemma.
\end{proof}

Given the limitations of Slater determinants  exposed by Lemma~\ref{lemma:slaterlimitation}, it is natural to consider approximation algorithms for more general classes of fermionic states.
A natural candidate is the class of fermionic Gaussian states.
The latter are defined most naturally in terms of Majorana fermion operators
$c_1,\ldots,c_{2n}$ such that 
\be
\label{Majoranas}
c_{2p-1}=a_p + a_p^\dag \quad \mbox{and} \quad c_{2p}=(-i)(a_p - a_p^\dag)
\ee
for $1\le p\le n$. They obey commutation rules $c_p^\dag= c_p$ and
$c_p c_q + c_q c_p = 2\delta_{p,q}I$.
For any orthogonal matrix  $R\in O(2n)$
let $U_R$ be a unitary operator acting on the Hilbert space
of $n$ fermi modes  such that
\be
\label{Rmatrix}
(U_R)^\dag \, c_p (U_R) = \sum_{q=1}^{2n} R_{p,q} c_q
\ee
for all $p=1,\ldots,2n$. Such unitary $U_R$ 
is uniquely determined by $R$ up to an overall phase~\cite{bravyi2004lagrangian}.
A state $\psi$ of $n$ fermi modes is called Gaussian if it has the form 
\be
|\psi\ra = U_R |0^n\ra
\ee
for some orthogonal matrix $R\in O(2n)$.
Slater determinants can be viewed as a subset of Gaussian states
such that the rotation Eq.~(\ref{Rmatrix}) does not mix creation and annihilation
operators. 
Given a fermionic Hamiltonian $h$, let 
\be
\label{LGauss}
\Lgauss{h}= \max_\psi \la \psi|h|\psi\ra,
\ee
where the maximization is over all Gaussian states $\psi$. Approximation algorithms based on Gaussian states 
as well as projections of Gaussian states onto a fixed particle 
number subspace have been previously used
as an extension of the Hartree-Fock method from quantum chemistry
\cite{bach1994generalized,tahara2008variational,kraus2010generalized,bravyi2017complexity}.
 We find the following:
\begin{lemma}\label{lemma:gaussrichardson}
Let $h$ be Richardson's Hamiltonian 
acting on $n\ge 8$ fermi modes.  Then
\be
\label{GaussLower}
\frac{\Lgauss{h}}{\Lmax{h}}\geq 1-\frac6n.
\ee
\end{lemma}
From Eqs.~(\ref{SlaterUpper},\ref{GaussLower}) 
we infer  that 
approximation ratios achieved by Gaussian and Slater determinant states
for  Richardson's Hamiltonian
 approach $1$ and $0$ respectively in the limit $n\to \infty$.  Thus general Gaussian states can vastly outperform  Slater determinants as a variational ansatz, 
even if one considers  particle number preserving Hamiltonians. 
We note that similar conclusions have previously been reached 
by Bach, Lieb, and Solovej~\cite{bach1994generalized},
as well as Kraus and Cirac~\cite{kraus2010generalized}
 in the study of the fermionic Hubbard model
with attractive interactions.
\begin{proof}[\bf Proof of Lemma~\ref{lemma:gaussrichardson}]
We shall need the following well-known fact, see e.g. Ref.~\cite{bravyi2006universal}.
\begin{fact}
\label{fact:paired}
Consider any permutation $\sigma \in S_{2n}$. 
Then there is a unique (up to an overall phase) state $\psi$ satisfying
\[
c_{\sigma(2p-1)} c_{\sigma(2p)} |\psi\ra = i|\psi\ra, \qquad p=1,\ldots,n.
\]
The state $\psi$ is Gaussian for any permutation $\sigma$.
\end{fact}
We claim that 
\be
\label{Lgauss_Lower}
\Lgauss{h} \ge \frac{N(N+1)}4.
\ee
Indeed,  rewrite $P$ in terms of Majorana operators 
\[
\alpha_j\equiv c_{4j-3}, \quad \beta_j\equiv c_{4j-2}, \quad \gamma_j\equiv c_{4j-1},\quad
\delta_j\equiv c_{4j}.
\]
Then $a_{2j-1}=(1/2)(\alpha_j-i\beta_j)$ and $a_{2j}=(1/2)(\gamma_j-i\delta_j)$ so that
\be
P=\frac{1}{4} \sum_{j=1}^{N}\left( \alpha_j \gamma_j-\beta_j\delta_j\right)
-i\left(\beta_j \gamma_j+\alpha_j\delta_j\right).
\label{eq:P}
\ee
Let $\psi$ be a ``paired'' state defined by 
\begin{equation}
\beta_j \gamma_j |\psi\rangle=i|\psi\rangle 
 \quad \text{and} \quad
\alpha_j \delta_j |\psi\rangle=i|\psi\rangle
\label{eq:stabs}
\end{equation}
for $1\leq j\leq N$. 
Note that Eq.~(\ref{eq:stabs}) contains $2N=n$ pairs of Majorana
operators and all pairs are disjoint. Thus Fact~\ref{fact:paired} implies that 
$\psi$ is  Gaussian.  
Using Eq.~\eqref{eq:P} gives
\begin{equation}
P|\psi\rangle=\frac{N}{2} |\psi\rangle+\frac{1}{4}\sum_{j=1}^{N}  \left(\alpha_j \gamma_j-\beta_j \delta_j\right)|\psi\rangle.
\label{eq:Pphi}
\end{equation}

From Eq.~\eqref{eq:stabs} we see that $\langle \psi |O|\psi\rangle=0$ for any operator $O$ which 
anticommutes with $\beta_j\gamma_j$ or $\alpha_j \delta_j$. Using this fact we obtain
\begin{align}
\langle \psi| \alpha_j \gamma_j \beta_k\delta_k|\psi\rangle&=\delta_{jk}\label{eq:s1}\\
\langle \psi| \alpha_j \gamma_j \alpha_k \gamma_k |\psi\rangle &=-\delta_{jk}\label{eq:s2}\\
\langle \psi|\beta_j \delta_j \beta_k \delta_k|\psi\rangle&=-\delta_{jk}\label{eq:s3}.
\end{align}
(Here $\delta_{j,k}$ denotes the Kronecker delta whereas $\delta_j$ denotes
a Majorana operator.)
Putting together Eqs.~(\ref{eq:Pphi}-\ref{eq:s3}) we arrive at
\begin{equation}
\Lgauss{h}\geq \langle \psi|h|\psi\rangle=\frac{N(N+1)}{4}.
\label{eq:gauss}
\end{equation}
Combining Eqs.~(\ref{Lmax},\ref{Lgauss_Lower}) proves the lemma.
\end{proof}

Next let us establish a lower bound  on the approximation ratio achieved by Gaussian
states for more general fermionic Hamiltonians that can 
be written in terms of quadratic and quartic Majorana operators:
\begin{align}
h&=h_1+h_2
\label{eq:h4}\\
h_1&=\sum_{p,q=1}^{2n} iV_{pq} c_pc_q \qquad h_2=\sum_{p,q,r,s=1}^{2n} W_{pqrs} c_p c_q c_r c_s.
\nonumber
\end{align}
Here $V_{p,q}$ and $W_{pqrs}$ are real coefficients. 
We shall assume that $V$ and $W$ are antisymmetric under a transposition of any
pair of indices. This guarantees that $h_1$ and $h_2$ are hermitian
and traceless. Otherwise, $V$ and $W$ can be completely arbitrary.
In particular, below we do not assume that $h$ is particle number preserving.
As before, the quadratic term $h_1$ is the fermionic analog of a $1$-local Hamiltonian. It is exactly solvable and all its eigenstates are Gaussian states~\cite{bravyi2004lagrangian}. 

By definition, 
this class of  Hamiltonians contains the models defined in Eq.~(\ref{fermiH12}).
Furthermore, it subsumes the
$2$-local qubit Hamiltonians considered in Section~\ref{sec:qubit}.
Indeed, suppose 
we are given a traceless $N$-qubit $2$-local Hamiltonian of the form Eq.~\eqref{eq:Htraceless}. 
Assume for simplicity that $N$ is even. 
Then we may  (efficiently) compute a Hamiltonian of the form Eq.~\eqref{eq:h4} with $n=3N$ fermi modes that has the same maximal eigenvalue $\lambda_{\max}(H)$ \cite{tsvelik1992new}. This proceeds by encoding each qubit $1\leq a\leq N$ using three Majorana modes $c_{3a-2}, c_{3a-1}, c_{3a}$ and representing the qubit Pauli operators as 
\[
X_{a}=i c_{3a-2}c_{3a-1} \quad Y_{a}=i c_{3a-1}c_{3a} \quad Z_{a}=i c_{3a-2}c_{3a}
\]
One can directly verify that they satisfy the correct Pauli commutation relations. Making this replacement for all Pauli operators in Eq.~\eqref{eq:Htraceless} we obtain a Hamiltonian of the form Eq.~\eqref{eq:h4}. One can show that this transformation preserves eigenvalues, while the degeneracy of each eigenvalue is increased by a factor of $2^{N/2}$ 
\footnote{Each encoded logical operator $X_{a}, Y_{a}, Z_{a}$ commutes with $b_a=ic_{3a-2}c_{3a-1}c_{3a}$ for each $1\leq a\leq N$ and $\{b_j,b_k\}=2\delta_{jk}$. The Hamiltonian acts on a Hilbert space $\cal{H}_{A}\otimes \cal{H}_B$, where $\cal{H}_A$ is the $N$-qubit system with logical operators $\{X_{a}, Y_{a}, Z_{a}: 1\leq a\leq N\}$, and $\cal{H}_B$ is a system of $N$ Majorana fermions $\{b_a: 1\leq a\leq N\}$. To complete the proof, note that the Hamiltonian acts trivially on $\cal{H}_B$, which has dimension $2^{n/2}$. }.

As in the case of $2$-local qubit Hamiltonians, the presence of the linear term $h_1$ in Eq.~\eqref{eq:h4} is a bit unwieldy and it suffices to consider Hamiltonians with quartic terms only. This (efficient) reduction  to the case $h_1=0$ proceeds using the following fermionic analogue of Lemma~\ref{lem:timereversal}. Given a Hamiltonian of the form Eq.~\eqref{eq:h4} we define a related Hamiltonian on $n+1$ fermi modes
\[
h^{\prime}=h_1\cdot\left(-i c_{2n+1}c_{2n+2}\right)+h_2
\]
Note that $h'$ contains only quartic Majorana operators.
\begin{lemma}
We have $\lambda_{\max}(h^{\prime})=\lambda_{\max}(h)$. Moreover, for any $(n+1)$-mode Gaussian 
or Slater determinant
state $\omega$ we may efficiently compute an $n$-mode Gaussian
or Slater determinant
 state $\phi$ such that $\langle \phi|h|\phi\rangle\geq \langle \omega| h^{\prime}|\omega\rangle$.
\label{lem:fermi_timereversal}
\end{lemma}
The proof of Lemma \ref{lem:fermi_timereversal},
provided in Appendix~\ref{app:fermi_timereversal},
is based on a fermionic analogue of the time reversal operation.
Finally, we establish a fermionic analogue of Theorem~\ref{thm:alg} for 
Gaussian states. 
\begin{theorem}
\label{thm:gauss_approx}
There is a classical algorithm which takes as input a Hamiltonian of the
form Eq.~\eqref{eq:h4} and outputs a Gaussian state $\psi$ such that 
\begin{equation}
\la \psi|h|\psi\ra \ge \frac{\Lgauss{h}}{O(\log{n})}
\label{eq:gaussianstate}
\end{equation}
and
\begin{equation}
\la \psi|h|\psi\ra \ge \frac{\Lmax{h}}{O(n\log{n})}
\end{equation}
with probability at least $2/3$.
The algorithm has runtime $poly(n)$. 
\end{theorem}
The proof is given in Appendices~\ref{app:SlaterProof},\ref{app:GaussLower}.
We leave as an open question whether Gaussian states achieve
a constant approximation ratio, that is, whether
\be
\label{conjecture}
\Lgauss{h} \ge C \Lmax{h}
\ee
for some universal constant $C>0$ 
and for all fermionic Hamiltonians $h$ of the form Eq.~(\ref{eq:h4}). If true, the conjecture Eq.~(\ref{conjecture}) would imply that the approximation algorithm of
Theorem~\ref{thm:gauss_approx} outputs a Gaussian state  $\psi$
with energy
$\la \psi |h|\psi\ra \ge \Lmax{h}/O(\log{n})$.
This would match the best known approximation algorithms
for classical quadratic functions and $2$-local qubit Hamiltonians,
see Sections~\ref{sec:intro},\ref{sec:qubit}.

\section{Acknowledgements}
The authors thank Boaz Barak, Sevag Gharibian, Aram Harrow, and 
Frank Verstraete for helpful discussions and comments.
 SB, DG and KT acknowledge support from the IBM Research
Frontiers Institute. RK is supported by the Technical University of Munich -- Institute for Advanced Study,
funded by the German Excellence Initiative and the
European Union Seventh Framework Programme under
grant agreement no.~291763.

\input{main.bbl}

\onecolumngrid

\appendix

\section{Lower bound on the maximum eigenvalue of $2$-local Hamiltonians}
\label{app:LmaxLower}

In this section we consider qubit Hamiltonians $H_2=\sum_{i,j=1}^{3n} C_{i,j} P_i P_j$
that contain only weight-two Pauli operators. 
\begin{lemma} [\cite{harrow2017extremal}]
\[
\lambda_{\max}(H_2)\geq \frac{\sum_{i,j=1}^{3n} |C_{i,j}|}{9n}
\]
\label{lem:weakbound}
\end{lemma}

\begin{proof}
Define a graph $G=(V,E)$ where $V=[n]$ and  $E$
contains nine edges connecting each pair of vertices $a\ne b\in V$.
One should think of vertices and edges of $G$ as qubits and 
two-qubit Pauli operators respectively. 
Then there is a two-to-one correspondence between
the Pauli terms of the Hamiltonian  $H_2$ and 
the edges of $G$ (for example, if $P_i=Z_1$ and $P_j=Z_2$,
we would count $P_iP_j$ and $P_j P_i$ as two different terms of $H_2$
whereas they are represented by the same edge of $G$).
 Thus one can write
\begin{equation}
\label{HH}
H_2=\sum_{i,j=1}^{3n}C_{i,j} P_i P_j=2\sum_{e\in E} C_e P_e,
\end{equation}
where $C_e\equiv C_{i,j}$ and $P_e\equiv P_i P_j$.

Assume for simplicity that $n$ is even. 
Let $\calM$ be the set of perfect matchings on $G$.
We claim that 
for any fixed $M\in \calM$ one can define a random  $n$-qubit state
$|\phi\rangle=|\phi_1\otimes \cdots \otimes \phi_n\rangle$ such that 
\begin{equation}
\label{match2}
\EE_\phi \langle \phi | C_e P_e |\phi\rangle = 
\left\{ \ba{rcl}
|C_e| &\mbox{if} & e\in M,\\
0 && \mbox{otherwise}\\
\ea
\right.
\end{equation}
Indeed, consider some fixed edge $e\in M$ and let
$a<b$ be the qubits connected by $e$.
Write $P_e=Q_a Q_b$, where $Q_a,Q_b\in \{X,Y,Z\}$.
Choose single-qubit Clifford gates $U_a,U_b$ such that 
\[
Q_a=U_a Z_a U_a^\dag \quad \mbox{and} \quad Q_b=U_b Z_b U_b^\dag.
\]
Let $x\in \{0,1\}$ be a random uniformly distributed  bit.  Set
\[
|\phi_a\rangle =U_a|x\rangle \quad \mbox{and} \quad |\phi_b\rangle = U_b|x\oplus y\rangle
\]
where $y=0$ if $C_e\ge 0$ and $y=1$ if $C_e<0$. 
Then 
\begin{equation}
\label{phi}
\EE_x \langle \phi_a \otimes \phi_b | Q_aQ_b|\phi_a\otimes \phi_b\rangle = \mathrm{sgn}(C_e)
\quad \mbox{and} \quad
\quad \EE_x |\phi_a\rangle\langle \phi_a| =  \EE_x |\phi_b\rangle\langle \phi_b|=\frac{I}2.
\end{equation}
Let us define random single-qubit states $\phi_a,\phi_b$ as above independently for every edge $e\in M$. 
Choosing $\phi$ as a tensor product of all $\phi_a$ and using 
Eq.~(\ref{phi}) one easily gets Eq.~(\ref{match2}).

Suppose now that $M\in \calM$ is picked at random from the uniform distribution. 
Let $\phi^M\equiv \phi$ be a random product state satisfying Eq.~(\ref{match2}). Then
\begin{equation}
\label{eta}
\EE_{M} \EE_{\phi^M}  \langle \phi^M | H_2|\phi^M\rangle = 2\sum_{e\in E} \mathrm{Pr}(e\in M) |C_e| =
\frac{2}{9(n-1)} \sum_{e\in E} |C_e| = \frac1{9(n-1)} \sum_{i,j=1}^{3n} |C_{i,j}|,
\end{equation}
where we noted that $\mathrm{Pr}(e\in M)=1/(9(n-1))$. 
This establishes the existence of a product state which achieves energy at least the right-hand-side.
\end{proof}

\section{Stronger version of Lieb's theorem}
\label{app:Lieb}

In this section we prove 
a lower bound $\|H\|_{\mathrm{sep}}/\|H\| \ge 1/6$, see Section~\ref{sec:qubit}.
\begin{proof}
In light of Lemma \ref{lem:timereversal} it suffices to consider the case $H_1=0$.

Let $\rho$ satisfy $\trace{\left(\rho H\right)}=\lambda_{\max}(H)$. Let $S\subset [n]$ be a uniformly random subset of size $|S|=\alpha n$ and let $T=[n]\setminus S$. Here $\alpha$ is a constant we will fix later.  Define
\[
\rho(S)=\left(\mathcal{E}_{1/3}^S\otimes I_T\right)(\rho)
\]
In other words $\rho(S)$ is obtained from $\rho$ by applying the entanglement-breaking depolarizing channel $\mathcal{E}_{1/3}$ to all qubits in $S$.  Write the Hamiltonian as
\[
H=H_S+H_T+H_{ST}
\]
where
\[
H_S=\sum_{i,j\in S} C_{ij} P_iP_j \qquad H_T=\sum_{i,j\in T} C_{ij} P_iP_j \qquad H_{ST}=H-H_S-H_T.
\]
Now we have
\begin{equation}
\mathbb{E}_S \left[\trace{\left(\rho(S) H_S\right)}\right]=\mathbb{E}_S \left[\trace{\left(\rho \cdot \frac{1}{9}\sum_{i,j\in S} C_{ij} P_iP_j\right)}\right]=\frac{\alpha^2 \lambda_{\max}(H)}{9}
\label{eq:ex1}
\end{equation}
where we used the fact that $\mathrm{Pr}\left[i\in S \text{ and } j\in S\right]=\alpha^2$ for any $i\neq j$. Similarly,
\begin{equation}
\mathbb{E}_S \left[\trace{\left(\rho(S) H_{ST}\right)}\right]=\frac{2\alpha(1-\alpha) \lambda_{\max}(H)}{3}.
\label{eq:ex2}
\end{equation}

Since $\mathcal{E}_{1/3}$ is entanglement-breaking we may write
\begin{equation}
\rho(S)=\sum_{k} p(k)\sigma^{k}_{S}\otimes \epsilon^{k}_{T}  \qquad \quad 
\label{eq:rhos}
\end{equation}
where for each index $k$, $\sigma^{k}_S$ is a product state of the qubits in $S$, and $p(k)$ is a probability distribution. Here $\epsilon^{k}_T$ is some state of the qubits in $T$ which is in general not separable. For each $k$ and $j\in T$ we denote the single-qubit marginals of $\epsilon^k$ as
\begin{equation}
\omega_{j}^{k} =\trace{_{T\setminus \{j\}}}(\epsilon^k) \qquad \quad j\in T.
\end{equation}
Since it is a one-qubit density matrix, we may write $\omega_{j}^{k}$ as 
\[
\omega_{j}^{k}=1/2(I+aX+bY+cZ).
\]
For each $j,k$ we extend the above to a one-parameter family
\[
\omega_{j}^{k} (t)=\omega_{j}^{k}=1/2(I+taX+tbY+tcZ) \qquad -1\leq t\leq 1.
\]
Finally, we define a family of separable states $\Gamma(S,t)$ which are obtained by replacing the state $\epsilon^k$ in Eq.~\eqref{eq:rhos} as follows
\[
\Gamma(S,t)=\sum_{k} p(k)\sigma^{k}_{S}\otimes \bigotimes_{j\in T} \omega_j^{k}(t).
\]
Note that 
\[
\mathbb{E}_S\left[\trace{\left(\Gamma(S,t)\cdot (H_S+H_{ST})\right)}=\trace{\left(\rho(S)\cdot (H_S+tH_{ST})\right)}\right]=\left(\frac{\alpha^2}{9}+t\frac{2\alpha(1-\alpha)}{3}\right)\lambda_{\max}(H)
\]
and we may write
\[
\mathbb{E}_S\left[\trace{\left(\Gamma(S,t)\cdot H_T\right)}\right]=t^2 F(\alpha) 
\]
where $F(\alpha) \in \mathbb{R}$ is a function of $\alpha$ defined by the left-hand side with $t=1$. We now fix $\alpha=1/2$. Putting together the above and writing $F=F(1/2)$ gives
\begin{equation}
\mathbb{E}_S\left[\trace{\left(\Gamma(S,t)H\right)} \right]=\left(\frac{1}{36}+\frac{t}{6}\right)\lambda_{\max}(H)+t^2F.
\label{eq:gamma}
\end{equation}

We now consider two cases depending on the value of $F$:
\newline
\newline
\textbf{Case 1: $F\geq -\frac{1}{36}\lambda_{\max}(H) $}
\newline
Choosing $t=1$ in Eq.~\eqref{eq:gamma} gives
\[
\lambda_{\mathrm{sep}}(H)\geq \mathbb{E}_S\left[\trace{\left(\Gamma(S,t=1)H\right)} \right]\geq \left(\frac{1}{36}+\frac{1}{6}-\frac{1}{36}\right)\lambda_{\max}(H)=\frac{1}{6}\lambda_{\max}(H).
\]
\newline
\newline
\textbf{Case 2: $F\leq -\frac{1}{36}\lambda_{\max}(H) $}
\newline
Plugging in $t=-1$ and multiplying Eq.~\eqref{eq:gamma} by $-1$ gives
\[
\lambda_{\mathrm{sep}}(-H)\geq -\mathbb{E}_S\left[\trace{\left(\Gamma(S,t=-1)H\right)} \right]\geq \left(\frac{1}{36}t^2-\frac{1}{36}-\frac{1}{6}t\right)_{t=-1}\lambda_{\max}(H) =\frac{1}{6}\lambda_{\max}(H).
\]
\newline
\newline
We have shown that either $\lambda_{\mathrm{sep}}(-H)\geq \frac{1}{6}\lambda_{\max}(H)$ or  $\lambda_{\mathrm{sep}}(H)\geq \frac{1}{6}\lambda_{\max}(H)$. This establishes that 
\begin{equation}
\|H\|_{\mathrm{sep}}\geq \frac{1}{6}\lambda_{\max}(H).
\label{eq:hnorm1}
\end{equation}
Applying Eq.~\eqref{eq:hnorm1} to $H$ and $-H$ we arrive at the statement of the theorem.
\end{proof}

\section{Fermionic time reversal operation}
\label{app:fermi_timereversal}

In this section we prove Lemma \ref{lem:fermi_timereversal}.
\begin{proof}
Below we shall make use of the complex conjugation and transpose operations with respect to the standard basis
$\{|x\ra\}$, $x\in \{0,1\}^n$.
Expressing Majorana operators $c_1,\ldots,c_{2n}$ in terms of qubit Pauli operators one gets 
\begin{align}
c_{1}& =X_1\qquad \qquad c_{2}=Y_1 \nonumber\\
c_{2j-1}&=Z_1\otimes \ldots Z_{j-1} \otimes X_j \qquad  c_{2j}=Z_1\otimes \ldots Z_{j-1} \otimes Y_j \qquad \qquad 2\leq j\leq n.
\end{align}
From the above we see that
\begin{equation}
c_{2j}^{\star}=c_{2j}^{T}=-c_{2j} \qquad \text{and} \qquad c_{2j-1}^{\star}=c_{2j-1}^{T}=c_{2j-1} \qquad 1\leq j\leq n.
\label{eq:cstar}
\end{equation}

Now let us proceed with the proof of Lemma \ref{lem:fermi_timereversal}. We may simultaneously diagonalize the commuting operators $h^{\prime}$ and $Z_{n+1}=-i c_{2n+1}c_{2n+2}$. The eigenvalues of $h^{\prime}$ consist of eigenvalues of $h_1+h_2$  ($Z_{n+1}=+1$) as well as all eigenvalues of $h_2-h_1$  ($Z_{n+1}=-1$), i.e., 
\begin{equation}
\lambda_{\max}(h^{\prime})=\max\left\{\lambda_{\max}(h_1+h_2), \lambda_{\max}(-h_1+h_2)\right\}.
\label{eq:same}
\end{equation}

Define a unitary
\[
U_{\mathrm{odd}}=c_{1}c_{3}c_{5}\ldots c_{2n-1}.
\]
A direct calculation using Eq.~\eqref{eq:cstar} gives
\begin{equation}
U_{\mathrm{odd}}^{\dagger} (h_1+h_2) U_{\mathrm{odd}}= (-h_1+h_2)^{T}.
\label{eq:oddconj}
\end{equation}
Since conjugation by the unitary matrix $U_{\mathrm{odd}}$ and the transpose operation both preserve the spectrum, we see that Hamiltonians $h_1+h_2$ and $-h_1+h_2$ have the same spectrum.
Thus so do $h$ and $h^{\prime}$. 

Now suppose that $\omega$ is an $(n+1)$-mode state which is either a Gaussian state
or a Slater determinant. Define a projector
\[
\Pi_{z}=\frac{1}{2}\left(1+(-1)^z Z_{n+1}\right).
\]
Since $Z_{n+1}$ commutes with $h^{\prime}$, one of the states
\[
|\omega(z)\rangle=\frac{1}{\|\Pi_{z}|\omega\rangle\|}\Pi_{z} |\omega\rangle, \qquad  \quad z\in \{0,1\}
\]
has energy at least $\langle \omega|h^{\prime}|\omega\rangle$. Note that we may write
\[
|\omega(z)\rangle=|\alpha\rangle\otimes |z\rangle.
\]
Now observe that
\[
\langle \alpha|h|\alpha\rangle=\langle \omega(0)|h^{\prime}|\omega(0)\rangle
\]
and, using Eq.~\eqref{eq:oddconj},
\[
\langle \alpha^{\star} |U_{\mathrm{odd}}^{\dagger}hU_{\mathrm{odd}}|\alpha^{\star}\rangle=\langle \alpha|-h_1+h_2|\alpha\rangle=\langle \omega(1)|h^{\prime}|\omega(1)\rangle
\]
The $n$-mode $\phi$ claimed in the Lemma is chosen to be $|\phi\rangle=|\alpha\rangle$ if  
\[
\langle \omega(0)|h^{\prime}|\omega(0)\rangle \geq \langle \omega|h^{\prime}|\omega\rangle 
\]
and $|\phi\rangle=U_{\mathrm{odd}}|\alpha^{\star}\rangle$ otherwise. 
It remains to notice that the projector $\Pi_z$ maps the set of Gaussian states
(Slater determinant states) to itself~\cite{bravyi2004lagrangian}.
The same applies to  the complex conjugation operation. 
Thus the state $\phi$ is Gaussian (Slater determinant) whenever $\omega$ is
Gaussian (Slater determinant).
\end{proof}

\section{Maximum eigenvalue of  Richardson's model}
\label{app:Richardson}

Recall that we consider $n=2N$ fermi modes and a  Hamiltonian 
\be
h=P^\dag P, \qquad P=\sum_{j=1}^N a_{2j-1} a_{2j}.
\ee
In this section we prove that the largest eigenvalue of $h$ is 
\be
\label{Lmax1}
\lambda_{\max}(h) = \frac{N(N+2) + \epsilon_N}4,
\ee
where $\epsilon_N={N \pmod 2} \in \{0,1\}$.

First let us prove that Eq.~(\ref{Lmax1}) gives an upper bound
on $\lambda_{\max}(h)$. Consider an identity decomposition
\[
I=\sum_{k=0}^n \Lambda_k
\]
where $\Lambda_k$ is a projector onto the $k$-particle subspace 
(that is, the subspace spanned by all states with exactly $k$ occupied modes). 
 Since $h$
is particle number preserving, one has $h\Lambda_k=\Lambda_k h$ for all $k$.
Therefore
\be
\lambda_{\max}(h)=\max_{0\le k\le n} \| \Lambda_k h \Lambda_k\| = \max_{0\le k\le n} \| P\Lambda_k\|^2,
\ee
By definition, each term in $P$ annihilates a pair of particles. Therefore
$P\Lambda_k=0$ for $k=0,1$
and $P\Lambda_k=\Lambda_{k-2} P\Lambda_k$ for $k\ge 2$. This gives
\be
\lambda_{\max}(h)=\max_{2\le k\le n} \| \Lambda_{k-2} P\Lambda_k\|^2.
\ee
We shall need the following simple fact.
\begin{fact}
\label{fact:sparse_norm}
Let $M$ be a complex matrix such that each row of $M$ has at most $R$ non-zeros
and each column has at most $C$ non-zeros. Then 
\be
\label{norm_fact}
\|M\| \le \sqrt{RC}\cdot \max_{i,j} |M_{i,j}|.
\ee
\end{fact}
For completeness, we provide a proof at the end of this section.
Choose $M=\Lambda_{k-2} P \Lambda_k$ and consider the matrix of $M$
in the standard  basis $\{|x\ra\}$, $x\in \{0,1\}^n$. Using the definition of $P$ one can check that 
each row of $M$ has at most $R=N+1-{\lceil k/2 \rceil}$ non-zeros
while each column of $M$ has at most $C={\lfloor k/2 \rfloor}$ non-zeros. Furthermore,
each non-zero element of $M$ has magnitude one. Thus
\[
\|\Lambda_{k-2} P \Lambda_k\|^2 \le (N+1-{\lceil k/2 \rceil}){\lfloor k/2 \rfloor}\equiv f(k).
\]
Therefore
\[
\lambda_{\max}(h)\le \max_{2\le k\le n}  f(k)=\frac{N(N+2) + \epsilon_N}4.
\]
To show that this upper bound is tight consider a state 
$|\psi_k\ra=(P^\dag)^k|0^n\ra$. Let 
\[
K=\sum_{j=1}^n a_j^\dag a_j
\]
be the particle number operator.
Using the commutation rules
$[h,P^\dag]=P^\dag(N\cdot I - K)$, 
one easily gets 
\[
h|\psi_k\ra=g(k)|\psi_k\ra, \quad \mbox{where} \quad g(k)=k(N-k+1).
\]
Thus 
\[
\lambda_{\max}(h)\ge \max_{0\le k\le N} g(k) = \frac{N(N+2) + \epsilon_N}4.
\]

\begin{proof}[\bf Proof of Fact~\ref{fact:sparse_norm}]
We shall label rows and columns of $M$ by $i$ and $j$ respectively. 
For each row $i$  let $\calC(i)$ be the set of columns $j$ such that
$M_{i,j}\ne 0$. For each column $j$ let $\calR(j)$ be the set of rows $i$
such that $M_{i,j}\ne 0$. Assume wlog that $|M_{i,j}|\le 1$ for all $i,j$.
Let $\psi$ be a normalized vector such that $\|M\|^2 = \| M \psi\|^2$. The Cauchy-Schwarz inequality gives
\[
\|M\psi\|^2 = \sum_i \left| \sum_{j\in \calC(i)} M_{i,j} \psi_j  \right|^2 
\le \sum_i  |\calC(i)| \cdot \sum_{j\in \calC(i)} |\psi_j|^2
\le R \sum_i \sum_{j\in \calC(i)} |\psi_j|^2.
\]
For the last inequality we used $|\calC(i)|\le R$. 
Changing the order of summations gives
\[
\|M\psi\|^2 \le R \sum_j |\psi_j|^2 \cdot |\calR(j)|\le  CR\sum_j |\psi_j|^2 = CR.
\]
\end{proof}

\section{Quadratic optimization with orthogonality constraints}
\label{app:QPOC}

Let $\calM_d$ be the space of real $d\times d$ matrices.
We say that $X\in \calM_d$ is positive semidefinite and write $X\ge 0$
if $X=X^T$ and all eigenvalues of $X$ are non-negative.
The notation $X\le Y$ stands for $Y-X\ge 0$.

Consider an objective function $F\, : \, \calM_d \to \RR$ such that
\be
\label{F(X)}
F(X)=\sum_{p,q,r,s=1}^d W_{pqrs} X_{p,q} X_{r,s} 
\ee
where $W_{pqrs}$  are real coefficients. We can assume wlog that 
$W_{pqrs}=W_{rspq}$
for all $p,q,r,s$. 
Consider the following problem~\cite{so2009improved}.
\begin{problem}[\bf Qp-Oc]
Given a linear subspace $\calL\subseteq \calM_d$,  
compute
\be
\label{QPOC}
\theta(\calL,W)=\max_{X\in \calL\, : \, \|X\|\le 1} \;  F(X).
\ee
\end{problem}
Here and below $\|X\|$ denotes the operator norm (the largest singular value).
Note that $X=0$ is a feasible solution so that $\theta(\calL,W)\ge 0$.
We shall need the following result established by So~\cite{so2009improved}.
\begin{lemma}
\label{lemma:So}
The problem Qp-Oc admits an approximation algorithm that 
outputs a matrix $X\in \calL$ such that 
$\|X\|\le 1$ and 
\be
F(X) \ge \frac{\theta(\calL,W)}{O(\log{d})}
\ee
with probability at least $2/3$.  The algorithm has runtime $poly(d)$.
\end{lemma}
We shall also need an SDP relaxation of Qp-Oc introduced in Ref.~\cite{so2009improved}.
To define this relaxation,  it will be convenient to identify a matrix
$X\in \calM_d$ and a vector  $|X\ra \in \RR^d\otimes \RR^d$ such that 
\[
|X\ra = \sum_{p,q=1}^d X_{p,q} |p,q\ra.
\]
Define an operator  $W\, : \, \RR^d\otimes \RR^d \to \RR^d\otimes \RR^d$ such that 
\[ 
W=\sum_{p,q,r,s=1}^d W_{pqrs} |p,q\ra\la r,s|.
\]
Note that $F(X)=\la X|W|X\ra$.
Suppose  $X\in \calM_d$ is a feasible solution of the Qp-Oc. 
Define 
\[
\rho=|X\ra\la X|=\sum_{p,q,r,s=1}^d X_{p,q} X_{r,s} |p,q\ra\la r,s|.
\]
The constraint $\|X\|\le 1$ gives $XX^T\le I$ and $X^TX\le I$.
The latter conditions can be rephrased in terms of the partial traces of $\rho$ as
\[
\mathrm{Tr}_1 (\rho)\equiv \sum_{p,q,s=1}^d X_{p,q} X_{p,s} |q\ra \la s| \le I
\qquad \mbox{and} \qquad 
\mathrm{Tr}_2 (\rho)\equiv \sum_{p,q,r=1}^d X_{p,q} X_{r,q} |p\ra \la r| \le I.
\]
Finally, let us choose a set of matrices $L^1,\ldots,L^k\in \calM_d$ such that 
\[
X\in \calL \quad \mbox{iff} \quad \la L^a |X\ra=0 \quad \mbox{for all $a=1,\ldots,k$}.
\]
Then clearly,  $\la L^a |\rho|L^a\ra=0$ for all $a$.
Ref.~\cite{so2009improved} defines the following  SDP relaxation of the Qp-Oc:
\be
\label{QPOCrelax}
\theta^*(\calL,W)=\max_{\rho}  \mathrm{Tr}(\rho W)
\quad \mbox{{\bf subject to}}
\quad 
\left\{ \ba{rcl}
\rho &\ge& 0,\\
\mathrm{Tr}_1 (\rho) & \le & I,\\
\mathrm{Tr}_2 (\rho) & \le & I,\\
\la L^a |\rho|L^a\ra & =& 0 \quad \mbox{for all $a=1,\ldots,k$}.\\
\ea\right.
\ee
Here the maximization is over symmetric real matrices $\rho$ of size $d^2\times d^2$, 
that is,
\[
\rho=\sum_{p,q,r,s=1}^d \rho_{pqrs} |p,q\ra\la r,s|, \qquad \rho_{pqrs}=\rho_{rspq}.
\]
Note that $\theta^*(\calL,W) \ge \theta(\calL,W)$ since 
Eq.~(\ref{QPOCrelax}) is a relaxation of Eq.~(\ref{QPOC}).
The following result is a special case of Theorem~1 from Ref.~\cite{so2009improved}.
\begin{lemma}
\label{lemma:So1}
\be
\theta^*(\calL,W) \le O(\log{d}) \cdot  \theta(\calL,W).
\ee
\end{lemma}
The approximation algorithm of Lemma~\ref{lemma:So} 
works by solving the SDP relaxation defined in Eq.~(\ref{QPOCrelax}) and 
representing  the optimal solution $\rho$ as 
a probabilistic mixture of pure states  $|X^\alpha\ra\la X^\alpha|$
such that $\|X^\alpha\| \le O(\sqrt{\log{d}})$ with high probability. Then $Y^\alpha\equiv X^\alpha/O(\sqrt{\log{d}})$ is a feasible solution of
the original Qp-Oc with the expected value of $F(Y^\alpha)$
equal to $\theta^*(\calL,W)/O(\log{d})$.

\section{Maximizing the energy over Slater determinants and Gaussian states}
\label{app:SlaterProof}

In this section we prove Theorem~\ref{thm:slater}
and the first part of Theorem~\ref{thm:gauss_approx}.
We begin by summarizing some well-known properties
of Gaussian and Slater determinant states, see e.g. 
Refs.~\cite{terhal2002classical,bravyi2004lagrangian}.

Suppose $\rho$ is a (mixed) state of $n$ fermi modes. 
We say that $\rho$ is a mixed Gaussian state iff there exists
an orthogonal matrix $R\in O(2n)$ 
and real numbers $\lambda_1,\ldots,\lambda_n \in [0,1]$
such that 
\be
\label{GaussMixed}
\rho = (U_R) \rho_\lambda (U_R)^\dag, \qquad \rho_\lambda\equiv \bigotimes_{j=1}^n 
(\lambda_j |0\ra\la 0| + (1-\lambda_j)|1\ra\la 1|).
\ee
Here $U_R$ is a unitary operator defined by Eq.~(\ref{Rmatrix}).
By definition, any mixed Gaussian state is a probabilistic mixture of pure ones, see Section~\ref{sec:fermion}.
Thus one can compute $\Lgauss{h}$ by maximizing the energy
$\mathrm{Tr}(\rho h)$ over the set of mixed Gaussian states $\rho$.

Given any $n$-mode state $\rho$,  define a covariance matrix
\[
X_{p,q} = (-i/2)\mathrm{Tr}(\rho (c_p c_q -c_q c_p)), \qquad 1\le p,q\le 2n.
\]
It is known that any mixed Gaussian state $\rho$ obeys Wick's theorem, that is,
\be
\label{WickCCCC}
-\mathrm{Tr}(\rho c_p c_q c_r c_s ) = X_{p,q} X_{r,s} - X_{p,r} X_{q,s} + X_{p,s} X_{q,r}
\ee
for any tuple~$(p,q,r,s)$ of pairwise distinct indices. 
Furthermore, for any real anti-symmetric matrix $X$ such that $\|X\|\le 1$ there exists
a mixed Gaussian state $\rho$ 
such that $X$ is the covariance matrix of $\rho$.
Combining 
the above facts gives
\be
\label{FermiQpOc}
\Lgauss{h}=\max_{X\in \calL\, : \, \|X\|\le 1}\; F(X),
\ee
where $\calL$ is the space of real anti-symmetric matrices of size $2n$
and
\be
\label{FermiF(X)}
F(X)=-3\sum_{p,q,r,s=1}^{2n} W_{pqrs} X_{p,q} X_{r,s}.
\ee
Here we used the assumption that $W$ is fully antisymmetric
and assumed wlog that $h_1=0$ (using Lemma~\ref{lem:fermi_timereversal}).
This is an instance of the Qp-Oc problem considered in the previous section.
The first part of Theorem~\ref{thm:gauss_approx} now follows directly
from Lemma~\ref{lemma:So}.

To prove Theorem~\ref{thm:slater} let us 
consider the subset of mixed Gaussian states $\rho$ satisfying
\be
\label{SlaterMixed}
\mathrm{Tr}(\rho\, a_p a_q)=0,\qquad  1\le p,q\le n.
\ee
We claim that 
such states are probabilistic mixtures of Slater determinant states. 
Indeed, define a covariance matrix $Q_{i,j}=\mathrm{Tr}(\rho a_i^\dag a_j)$.
Let $V$ be unitary operator such that $V^\dag Q V$ is a diagonal matrix
with entries $\lambda_1,\ldots,\lambda_n$ on the main diagonal. 
Define a new set of annihilation operators
\[
b_p=\sum_{q=1}^n V_{q,p} a_q, \qquad 1\le p\le n.
\] 
Then $\mathrm{Tr}(\rho b_p^\dag b_q)=\lambda_p \delta_{p,q}$.
Furthermore, Eq.~(\ref{SlaterMixed}) gives
$\mathrm{Tr}(\rho b_p b_q)=0$ for all $p,q$.
Wick's theorem now
implies that $\rho$ is a product of single-mode states
$\lambda_j b_j^\dag b_j+ (1-\lambda_j) b_j b_j^\dag$.
Thus $\rho$ is a mixture of pure product states such that 
each mode $b_j$ either empty or occupied, that is, 
$\rho$ is a mixture of Slater determinants, as claimed. 

Combining  the above facts gives
\be
\label{SlaterOpt}
\Lslater{h}=\max_{X\in \calL'\, : \, \|X\|\le 1}\; F(X),
\ee
where $\calL'\subseteq \calL$ is the linear subspace
of covariance matrices $X$ 
such that the corresponding mixed Gaussian state
satisfies Eq.~(\ref{SlaterMixed}).
Recalling that $a_p=(c_{2p-1} + i c_{2p})/2$ one gets
\[
\calL'= \{ X \in \calL \, : \,  X_{2p-1,2q-1}=X_{2p,2q} \quad \mbox{and} \quad
X_{2p-1,2q}=-X_{2p,2q-1} \quad \mbox{for all $1\le p <q\le n$}\}.
\]
We conclude that Eq.~(\ref{SlaterOpt}) is another instance of 
the Qp-Oc problem from the previous section.
Theorem~\ref{thm:slater} now follows directly
from Lemma~\ref{lemma:So}: the latter algorithm produces a mixed Gaussian state~$\rho$ satisfying~\eqref{SlaterMixed} and $\tr(\rho h)\geq \frac{\Lslater{h}}{O(\log n)}$.  Eq.~(\ref{GaussMixed}) shows that $\rho$ is a convex combination of Slater determinants. By first computing $U_R$ and the numbers $\lambda_1,\ldots,\lambda_n\in [0,1]$ from~$\rho$, one can thus efficiently sample from an ensemble of Slater determinant states~$\phi$ such that the expected value of $\bra{\phi}h\ket{\phi}$ is $\tr(\rho h)$. Repeatedly sampling from this distribution $O(\log(n))$ times and choosing the output Slater determinant state~$\phi$ with the highest energy, we obtain a success probability of at least $2/3$. This follows from an argument identical to that given after Eq.~\eqref{eq:expectationbound}.

As a side remark we note that optimization over separable states
considered in Section~\ref{sec:qubit} can  also be represented as a special case
of Qp-Oc (although this representation is slightly more cumbersome
and not very insightful).

\section{Proof of Theorem~\ref{thm:gauss_approx}}
\label{app:GaussLower}

Consider a fermionic Hamiltonian
\[
h=\sum_{p,q,r,s=1}^{2n} W_{pqrs} c_p c_q c_r c_s.
\]
Here we assumed wlog that $h_1=0$ (using Lemma~\ref{lem:fermi_timereversal}).
Let $\psi$ be a  largest eigenvector of $h$ such that 
\[
\Lmax{h}=\la \psi|h|\psi\ra.
\]
In this section we prove that $\Lgauss{h}\ge \Lmax{h}/O(n\log{n})$.
Moreover, we show that  the algorithm from the previous section outputs
a Gaussian state that achieves this approximation ratio. 
The proof given below depends crucially on the material
of Appendices~\ref{app:QPOC},\ref{app:SlaterProof}.

Define
an operator $\rho \, : \, \RR^{2n} \otimes \RR^{2n} \to \RR^{2n} \otimes \RR^{2n}$ such that
\be
\label{Yfeasible}
\rho=-\frac1{2n} \sum_{p,q,r,s=1}^{2n} \epsilon_{p,q} \epsilon _{r,s} \mathrm{Re}(\la \psi | c_p c_q c_r c_s|\psi\ra) |p,q\rangle\langle r,s|,
\ee
where $\epsilon_{p,q}=1$ if $p\ne q$ and $\epsilon_{p,q}=0$ if $p=q$.
Define also
an operator 
\[ 
W=-\sum_{p,q,r,s=1}^{2n} W_{pqrs} |p,q\ra\la r,s|.
\]
Here we introduced the minus sign to cancel the minus
sign that comes from Wick's theorem, see Eqs.~(\ref{FermiQpOc},\ref{FermiF(X)}).
Note that 
\be
\label{energy}
\mathrm{Tr}(\rho W)=\frac1{2n} \la \psi|h|\psi\ra = \Lmax{h}/2n.
\ee
We claim that $\rho$ is a feasible solution
of the SDP relaxation defined in Eq.~(\ref{QPOCrelax}),
where $d=2n$ and $\calL$ is the space of real anti-symmetric matrices.
Indeed, consider any state $|\phi\ra \in \RR^{2n} \otimes \RR^{2n}$. Define an operator
\[
O = \sum_{r,s=1}^{2n} \epsilon_{r,s} \la r,s |\phi\ra  c_r c_s.
\]
Taking into account that $\epsilon_{r,s} c_r c_s$ is anti-hermitian for all $r,s$ one gets
\[
\la \phi |\rho |\phi\ra =\frac1{2n} \la \psi |O^\dag O|\psi \ra \ge 0.
\]
Thus $\rho$ is positive semidefinite. Next,
\[
\mathrm{Tr}_1 (\rho) = -\frac1{2n} \sum_{p,q,s=1}^{2n} \epsilon_{p,q} \epsilon_{p,s}
\mathrm{Re}(\la \psi|c_p c_q c_p c_s|\psi\ra) \, |q\ra \la s|.
\]
Since $\epsilon_{p,q} c_p c_q= - \epsilon_{p,q} c_q c_p$ for all $p,q$ we arrive at
\[
\mathrm{Tr}_1 (\rho) = \frac1{2n} \sum_{p,q,s=1}^{2n} \epsilon_{p,q} \epsilon_{p,s}
\mathrm{Re}(\la \psi| c_q c_s|\psi\ra) \, |q\ra \la s|
=\frac{(2n-1)}{2n} I \le I.
\]
Here we noted that $\mathrm{Re}(\la \psi| c_q c_s|\psi\ra)=0$
for $q\ne s$ and $\mathrm{Re}(\la \psi| c_q c_s|\psi\ra)=1$
for $q=s$. By symmetry, $\mathrm{Tr}_2 (\rho)  \le I$.
Finally, $\rho$ has support
on the subspace $\calL$ of real anti-symmetric matrices
since the tensor
\[
\epsilon_{p,q} \epsilon _{r,s} \mathrm{Re}(\la \psi | c_p c_q c_r c_s|\psi\ra)
\]
is antisymmetric under the swap of $p,q$ and the swap of $r,s$.
This proves that $\rho$ is a feasible solution of the 
SDP relaxation Eq.~(\ref{QPOCrelax}) and thus 
\[
\Lmax{h} =2n \cdot \mathrm{Tr}(\rho W) \le 2n \cdot \theta^*(\calL,W)
\le O(n\log{n}) \cdot \theta(\calL,W) = O(n\log{n})\cdot \Lgauss{h}.
\]
Here we used Lemma~\ref{lemma:So1} and Eqs.~(\ref{FermiQpOc},\ref{FermiF(X)}).
It remains to note that the approximation algorithm
of Lemma~\ref{lemma:So} outputs
a feasible solution $X\in \calL$, $\|X\|\le 1$ such that $F(X)\ge \theta^*(\calL,W)/O(\log{(n)})$. 
Such $X$ defines a covariance matrix of a mixed Gaussian state
$\rho$ such that 
\[
\mathrm{Tr}(\rho\, h) \ge \frac{\Lmax{h}}{O(n\log{n})}.
\]
By definition, $\rho$ is a probabilistic mixture of pure Gaussian states
$\phi$ such that the expected value of $\la \phi|h|\phi\ra$ 
is $\mathrm{Tr}(\rho\, h)$. Repeatedly sampling from the corresponding distribution and choosing the state with the highest energy, we can amplify the success probability above~$2/3$ following the argument after Eq.~\eqref{eq:expectationbound}.  This proves Theorem~\ref{thm:gauss_approx}.

\end{document}

%% file: main.bbl
%